\UseRawInputEncoding
\documentclass[a4paper,onecolumn,11pt]{quantumarticle}
\pdfoutput=1
\usepackage[utf8]{inputenc}
\usepackage[T1]{fontenc}
\usepackage[draft = false, pagebackref]{hyperref}
\usepackage{amsmath}
\usepackage{amsthm}
\usepackage{thmtools}
\usepackage{amsfonts}
\usepackage{physics}
\usepackage{enumitem}
\usepackage{mmacells}

\usepackage{silence}
\DeclareFontFamily{U}{mathx}{\hyphenchar\font45}
\DeclareFontShape{U}{mathx}{m}{n}{
	<5> <6> <7> <8> <9> <10>
	<10.95> <12> <14.4> <17.28> <20.74> <24.88>
	mathx10
}{}
\DeclareSymbolFont{mathx}{U}{mathx}{m}{n}
\DeclareMathSymbol{\bigtimes}{1}{mathx}{"91}

\usepackage{ifdraft}
\ifdraft{
	
}{
	
}

\newtheorem{theorem}{Theorem}[section]
\newtheorem{prop}[theorem]{Proposition}
\newtheorem{lemma}[theorem]{Lemma}

\begin{document}
	
\title{On variants of multivariate quantum signal processing and their characterizations}

\author[1]{Balázs Németh}
\email{bn273@cam.ac.uk}
\affiliation[1]{Department of Applied Mathematics and Theoretical Physics, University of Cambridge, United Kingdom}
\author[2]{Blanka Kövér}
\author[2]{Boglárka Kulcsár}
\author[2]{Roland Botond Miklósi}
\affiliation[2]{Institute of Mathematics, Faculty of Science, Eötvös Loránd University, Budapest, Hungary}
\author[3]{András Gilyén}
\email{gilyen@renyi.hu}
\affiliation[3]{Alfréd Rényi Institute of Mathematics, HUN-REN, Budapest, Hungary}

\begin{abstract}
Quantum signal processing (QSP) is a highly successful algorithmic primitive in quantum computing which leads to conceptually simple and efficient quantum algorithms using the block-encoding framework of quantum linear algebra. Multivariate variants of quantum signal processing (MQSP) could be a valuable tool in extending earlier results via implementing multivariate (matrix) polynomials. However, MQSP remains much less understood than its single-variate version lacking a clear characterization of ``achievable'' multivariate polynomials. We show that Haah's characterization of general univariate QSP can be extended to homogeneous bivariate (commuting) quantum signal processing. We also show a similar result for an alternative inhomogeneous variant when the degree in one of the variables is at most 1, but construct a counterexample where both variables have degree 2, which in turn refutes an earlier characterization proposed / conjectured by Rossi and Chuang for a related restricted class of MQSP. Finally, we describe homogeneous multivariate (non-commuting) QSP variants that break away from the earlier two-dimensional treatment limited by its reliance on Jordan-like decompositions, and might ultimately lead to the development of novel quantum algorithms.
\end{abstract}

\maketitle

\section{Introduction}

Quantum signal processing~\cite{low2016CompositeQuantGates,low2016HamSimQSignProc} is a highly successful algorithmic primitive in quantum computing, which allows transforming the eigenvalues of potentially exponentially large Hermitian matrices~\cite{low2016HamSimQubitization,low2017HamSimUnifAmp} or more generally the singular values of arbitrary matrices~\cite{gilyen2018QSingValTransf} with efficient quantum circuits utilising only one or two ancilla qubits. These techniques are shown to generalize and achieve close-to optimal implementation of many important quantum algorithms~\cite{gilyen2018QSingValTransf} making the technique an essential tool in designing efficient quantum algorithms~\cite{dalzell2023QuantumAlgSurvey}. 

The matrices should be represented by a so-called block-encoding unitary on the quantum computer, whose top-left corner is the input/output matrix~\cite{low2016HamSimQubitization,chakraborty2018BlockMatrixPowers}. Such matrices can then be added by the linear combination of unitaries (LCU) technique~\cite{childs2012HamSimLCU}, and can also be multiplied with ease~\cite{chakraborty2018BlockMatrixPowers,gilyen2018QSingValTransf}. These techniques together imply that arbitrary (multivariate) polynomials of the input matrices can be implemented efficiently with a number of uses of the input block-encodings that roughly matches the degree of the polynomial. However, since the matrices are represented by block-encoding unitaries, the output matrices must be normalised, and naively chaining the above operations can lead to serious sub-normalisation. The sub-normalisation later can only be fixed at a cost that is roughly linear in the amplification factor~\cite{low2017HamSimUnifAmp,gilyen2018QSingValTransf}, so it is crucial to avoid any unnecessary sub-normalisation in constructing efficient quantum algorithms.

Quantum signal processing~\cite{low2016HamSimQSignProc} (QSP) and its extension to quantum singular value transformation~\cite{gilyen2018QSingValTransf} (QSVT) typically only introduce sub-normalisation that is required for making the output matrix have operator norm at most $1$ so that it can be block-encoded. However, if the separate terms of a multivariate polynomial are plainly summed via LCU, then much unnecessary sub-normalisation may be introduced. The induced overhead could be significant and may be proportional to the number of terms -- which could be very large. Moreover, LCU might require several ancilla qubits, while QSP and QSVT typically only introduces one or two additional ancilla qubits.

Multivariate quantum signal processing (MQSP)~\cite{rossi2022multivariableQSP} is hoped to resolve these issues in some cases by directly combining the terms of a multivariate polynomial by utilising a circuit analogous to ordinary quantum signal processing, where multiple different input matrices are used in predetermined locations instead of a single input matrix. A major barrier to applying MQSP is that the set of ``achievable'' multivariate polynomial transformations is severely limited, moreover no good characterization was known for them. We provide the first such characterization for the case of homogeneous bivariate (commuting) quantum signal processing, which can be viewed as a common generalization of the univariate characterizations of \cite{haah2018ProdDecPerFuncQSignPRoc,motlagh2023GeneralQSP}. We also describe general homogeneous multivariate (non-commuting) variants of MQSP, but leave it as an open question whether the necessary conditions that straightforwardly generalize those of the (commuting) bivariate case are also sufficient for guaranteeing the existence of a product decomposition.

A similar characterization was attempted for an alternative inhomogeneous bivariate version of MQSP in~\cite{rossi2022multivariableQSP}, relying on a conjecture about the structure of the appearing multivariate polynomials. In this paper we prove that their conjecture holds for bivariate polynomials which have degree at most $1$ in one of the variables, but does not hold in general, as we demonstrate by an explicitly constructed $2+2$-degree polynomial.
In fact, we consider a more general MQSP protocol then the one described in~\cite{rossi2022multivariableQSP}, following the generic treatment of~\cite{haah2018ProdDecPerFuncQSignPRoc,chao2020FindingAngleSequences} similar to \cite{motlagh2023GeneralQSP}. This makes our results more general, but also applicable to the restricted version studied in~\cite{rossi2022multivariableQSP}.
On top of analytical proofs, we developed a Python program package (see \autoref{apx:program}) to aid finding counterexamples, which we also made publicly available~\cite{gitHubRepo2023}. 

\section{Review of the univariate case}

We follow the presentation of~\cite{haah2018ProdDecPerFuncQSignPRoc,chao2020FindingAngleSequences} in our description of quantum signal processing.
We assume access\footnote{As discussed in~\cite{haah2018ProdDecPerFuncQSignPRoc}, if one only has access to controlled-$W$ (and its inverse), it can still be viewed as access to $V=\ketbra{0}{0}\otimes \sqrt{W}+\ketbra{1}{1}\otimes \sqrt{W}^{\dagger}$, in case one uses $V$ an even number of times. Therefore, arbitrary parity Laurent polynomials of $W$ might be implemented by using controlled-$W$ and its inverse.} to a black-box unitary $V:=\ketbra{0}{0}\otimes W+\ketbra{1}{1}\otimes W^{\dagger}$. 
Let us consider a fixed eigenvector $\ket{\theta}$ of $W$ with eigenvalue $t=e^{i\theta}\in\mathrm{U}(1)$, then the subspace spanned by $\ket{0}\ket{\theta}$, $\ket{1}\ket{\theta}$ is invariant under $V$.
Focusing on this subspace, one can show that the following
unitary transformations of the input signal $t$ can be implemented within this subspace, by acting on the first qubit with modulation unitaries $U_i$:
\begin{equation}
F(t)=U_{0}\begin{pmatrix}t & 0\\
0 & t^{-1}
\end{pmatrix}U_{1}\cdot\ldots\cdot U_{d-1}\begin{pmatrix}t & 0\\
0 & t^{-1}
\end{pmatrix}U_{d}\text{, where }U_{0},\dots,U_{d}\in\mathrm{SU}(2).\label{eq:univ_laur}
\end{equation}
Since $t$ represents an arbitrary eigenvalue of $W$, we can think of $F(t)$ as a matrix-valued Laurent polynomial in
the indeterminate $t$ having three important properties:
\begin{enumerate}
\item $\deg F(t)\leq d$, where the degree of a Laurent polynomial is the
maximum of the degrees of $t$ and $t^{-1}$ in it. \label{it:prop1}
\item $F(t)\in\mathrm{SU}(2)\quad\forall t\in\mathrm{U(1)}$.
\item $F(-t)=(-1)^d F(t)$.\label{it:prop3}
\end{enumerate}
Conversely, if $F(t)$ is a matrix-valued Laurent polynomial in $t$ satisfying the above properties, 
then $\exists U_{0},\dots,U_{d}\in\mathrm{SU}(2)$ such that \eqref{eq:univ_laur} holds, as shown by~\cite[Theorem~2]{haah2018ProdDecPerFuncQSignPRoc}.
This converse result is what renders quantum signal processing extremely powerful: 
it suffices to find an approximating polynomial for any transformation desired
to be applied, and it is guaranteed that the polynomial can be efficiently implemented using simple quantum circuits if the above properties \ref{it:prop1}-\ref{it:prop3} are satisfied. 
Increasingly efficient as well as numerically stable algorithms that are able to compute the $U_{i}$ unitary matrices have also been found~\cite{haah2018ProdDecPerFuncQSignPRoc,dong2020efficientPhaseFindingInQSP,chao2020FindingAngleSequences,wang2021energyLandscapeSymmetricQSP,motlagh2023GeneralQSP}.

Some further remarks about \eqref{eq:univ_laur}:
\begin{itemize}
\item Observe that for a given degree-$d$ Laurent polynomial $F(t)$ satisfying the above three
properties, the decomposition \eqref{eq:univ_laur} is not unique
because diagonal elements of $\mathrm{SU}(2)$ commute. To overcome
this ambiguity, Haah in \cite{haah2018ProdDecPerFuncQSignPRoc} has proposed the following product decomposition instead:
\begin{equation}\label{eq:haah_decomp}
F(t)=E_{0}E_{P_{1}}(t)E_{P_{2}}(t)\cdot\dots\cdot E_{P_{d}}(t),
\end{equation}
 where $E_{0}\in\mathrm{SU}(2)$, the $P_{i}$-s are rank-one projection matrices, 
 and the so-called \textit{primitive matrices} are defined as:
\begin{equation}
	E_{P_i}(t)=tP_i+t^{-1}(I-P_i)\in\mathrm{SU}(2).
\end{equation}The decomposition in \eqref{eq:haah_decomp} is unique, moreover, there is a simple recursive expression for the projection matrices $P_i$ in terms of the coefficients of the polynomial $F(t)$: let
\begin{equation}
	F(t)=\sum_{j=-d}^dC_jt^j,
\end{equation}
where the $C_j$ are $2\times 2$ complex matrices with $C_d\neq 0$, then $P_d=C_d^\dagger C_d/\mathrm{Tr}\left[C_d^\dagger C_d\right]$.
\item In many applications, it is of interest to restrict attention to certain more structured products than \eqref{eq:univ_laur}. For example, one might assume that $U_k=e^{i\phi_k\sigma_X}$ for $\phi_k\in\mathbb{R}$ and $\sigma_X=\begin{bmatrix}
	0&1\\1&0
\end{bmatrix}$~\cite{low2016CompositeQuantGates,gilyen2018QSingValTransf}.\footnote{In \cite{gilyen2018QSingValTransf} the ``signal'' is assumed to be a Pauli-$X$ rotation and the modulation is in the form of Pauli-$Z$ rotations, whereas here role of the Pauli-$Z$ and $X$ rotations are reversed. Nevertheless, the two treatments are equivalent up to a change of basis.} In this case, the entries of the matrix polynomial $F(t)$ (viewed as Laurent polynomials in $t$ with complex coefficients) obey additional symmetries. Let
\begin{align}
    \begin{split}
        F(t)&=e^{i\phi_0\sigma_X}\begin{pmatrix}
		t&0\\0&t^{-1}\end{pmatrix}e^{i\phi_1\sigma_X}\cdot\ldots\cdot e^{i\phi_{d-1}\sigma_X}\begin{pmatrix}
		t&0\\0&t^{-1}\end{pmatrix}e^{i\phi_d\sigma_X}\\
        &=\begin{pmatrix}
		P(t)&Q(t)\\
		-Q^*\left(t^{-1}\right)&P^*\left(t^{-1}\right)
	\end{pmatrix},\label{eq:struct_prod}
    \end{split}
\end{align}

then by induction on $d$ one can show that
\begin{equation}\label{eq:parities}
    P(t^{-1})=(-1)^dP\left(t\right),\quad Q(t^{-1})=(-1)^{d-1}Q\left(t\right).
\end{equation}
In turn, the satisfaction of these additional symmetries \eqref{eq:parities} and the above given properties \ref{it:prop1}-\ref{it:prop3} are sufficient to ensure that the corresponding $F(t)$ can be written in the form \eqref{eq:struct_prod} for some angles $\phi_0,\dots,\phi_d\in\mathbb{R}$~\cite{gilyen2018QSingValTransf,haah2018ProdDecPerFuncQSignPRoc}.

\end{itemize}

\section{Problem formulation of the multivariate homogeneous case}
A possible bivariate generalization of QSP is to consider what happens when one has access to a black-box unitary $V:=\ketbra{0}{0}\otimes W_1+\ketbra{1}{1}\otimes W_2$. Assuming that $W_1$ and $W_2$ commute, we once again consider a fixed common eigenvector $\ket{\theta}$ with corresponding eigenvalues $a=e^{i\theta}$ and $b=e^{i\varphi}$ respectively. The input signal then becomes a pair $(a,b)\in \mathrm{U(1)}\times\mathrm{U(1)}$ and the following polynomial transformations can be realized by repeated calls to the black-box unitaries modulated with fixed unitary gates on the first qubit:

\begin{equation}\label{eq:laur_bi_hom}
		F(a,b)=U_0\begin{pmatrix}a&0\\0&b\end{pmatrix}U_1
		\cdot\ldots\cdot
		U_{d-1}\begin{pmatrix}a&0\\0&b\end{pmatrix}U_d
		\text{, where }U_i\in\mathrm{SU(2)}.
\end{equation}

$F(a,b)$ is now a homogeneous bivariate polynomial in the variables $a,b$ characterized by the properties summarized in the theorem below, providing a common generalization to the results of~\cite{haah2018ProdDecPerFuncQSignPRoc,motlagh2023GeneralQSP}:

\begin{theorem}[Characterization of homogeneous commutative bivariate QSP]\label{thm:hombichar}
	The product Eq.~\eqref{eq:laur_bi_hom} describes a matrix-valued bivariate polynomial $F(a,b)$ such that

	\begin{enumerate}[label=(\roman*)]
		\item \begin{math}
			F(a,b)\text{ is a homogeneous degree-}d \text{ polynomial}. \label{bihompr1}
		\end{math}
		\item $F(a,b)\in\mathrm{U(2)}\quad\forall (a,b)\in\mathrm{U(1)}\times\mathrm{U(1)}.\label{bihompr2}$
		\item $\mathrm{det}(F(a,b))=(ab)^d.\label{bihompr3}$
	\end{enumerate}
	Conversely, any polynomial $F(a,b)$ satisfying \ref{bihompr1}-\ref{bihompr3} can be written in the from of Eq.~\eqref{eq:laur_bi_hom}.
\end{theorem}
\begin{proof}
	``$\Longrightarrow\colon$''
	\begin{enumerate}
	\item[\ref{bihompr1}] $U_i$ are constants for all $i\in\{0,1,\ldots, d\}$, while $\mathrm{diag}(a,b)$ has homogeneous degree $1$.
	\item[\ref{bihompr2}] $U_i\in \mathrm{SU(2)}$ for all $i\in\{0,1,\ldots, d\}$, while $\mathrm{diag}(a,b)\in\mathrm{U(2)}\quad\forall (a,b)\in\mathrm{U(1)}\times\mathrm{U(1)}$.
	\item[\ref{bihompr3}] Follows from the multiplicativity of determinants.		
\end{enumerate}
``$\Longleftarrow\colon$'' Consider $\tilde{F}(\sqrt{a/b},\sqrt{b/a}):=F(a,b)/(ab)^{\frac{d}{2}}$ as a degree $\leq d$ Laurent polynomial in the indeterminate $t:=\sqrt{a/b}$. Clearly, $\tilde{F}$ satisfies Properties \ref{it:prop1}-\ref{it:prop3} by construction and as a consequence of \ref{bihompr1}-\ref{bihompr3}. Thus $\exists U_{0},\dots,U_{d}\in\mathrm{SU}(2)$ such that Eq.~\eqref{eq:univ_laur} holds for $\tilde{F}(t,t^{-1})$, as shown by~\cite[Theorem~2]{haah2018ProdDecPerFuncQSignPRoc}. Finally, we can obtain $F(a,b)=(ab)^{\frac{d}{2}}\tilde{F}(\sqrt{a/b},\sqrt{b/a})$ using the same sequence of unitaries $U_{0},\dots,U_{d}$ as can be seen after replacing the matrices 
\begin{align*}
	\begin{pmatrix}t&0\\0&t^{-1}\end{pmatrix} \quad \text{ by } \quad \sqrt{ab}\begin{pmatrix}t&0\\0&t^{-1}\end{pmatrix}=\begin{pmatrix}a&0\\0&b \end{pmatrix}. \tag*{\qedhere}
\end{align*}
\end{proof}

One can view the above as a common generalization of the univariate characterizations of \cite{haah2018ProdDecPerFuncQSignPRoc,motlagh2023GeneralQSP}, however it might be more accurate to think about \cite[Theorem~2]{haah2018ProdDecPerFuncQSignPRoc}, \cite[Theorem 3]{motlagh2023GeneralQSP}, and  \autoref{thm:hombichar} as roughly equivalent statements, since the non-trivial ``$\Longleftarrow$'' direction of each can be used to prove that of the others using a reduction similar to the above proof.

\subsection{The non-commuting multivariate case}

We can ask the question what happens when the matrices $W_1$ and $W_2$ do not commute or we have more than two input unitaries $W_i$. In the non-commuting case we should interpret Eq.~\eqref{eq:laur_bi_hom} as a homogeneous bivariate polynomial in non-commuting variables $a,b$, meaning that we should separately treat the coefficients of expressions like $ab$ and $ba$. Similarly if we add more variables we get

\begin{equation}\label{eq:multi_hom}
	F(x_1,\ldots,x_n)=U_0\mathrm{diag}(x_1,\ldots,x_n)U_1
	\cdot\ldots\cdot
	\mathrm{diag}(x_1,\ldots,x_n)U_d\text{, where }U_i\in\mathrm{SU(n)}.
\end{equation}

In the non-commuting case we should interpret $F(x_1,\ldots,x_n)$ as a matrix-valued polynomial in the non-commuting variables $x_1,\ldots,x_n$, i.e., 
\begin{equation*}
	F(x_1,\ldots,x_n)=\sum_{j\in [n]^d}C_j \prod_{i=1}^d x_1^{\delta_{1j_i}}\cdot\ldots\cdot x_n^{\delta_{nj_i}}\text{, where }C_j\in\mathbb{C}^{n\times n}\text{ and }\delta_{\ell k}\text{ is the Kronecker-}\delta.
\end{equation*}
We can substitute operators $X_1,\ldots, X_n \in \mathrm{End}(\mathcal{H})$ of a Hilbert space $\mathcal{H}$ into such a polynomial as follows:
\begin{equation*}
	F(X_1,\ldots,X_n)=\sum_{j\in [n]^d}C_j \otimes\prod_{i=1}^d X_1^{\delta_{1j_i}}\cdot\ldots\cdot X_n^{\delta_{nj_i}}.
\end{equation*}

In the above $\mathrm{End}(\mathcal{H})$ denotes the set of endomorphisms of $\mathcal{H}$. Let $\mathrm{U}(\mathcal{H})\subset \mathrm{End}(\mathcal{H})$ denote the set of unitary operators on $\mathcal{H}$.
 In the non-commuting case we get necessary conditions analogous to \autoref{thm:hombichar}. 

\begin{theorem}[Necessary conditions on multivariate homogeneous non-commuting QSP]\label{thm:noncommhombichar}
	The product Eq.~\eqref{eq:multi_hom} describes a matrix-valued multivariate (non-commutative) polynomial $F(x_1,\ldots,x_n)$ such that
	
	\begin{enumerate}[label=\roman*)]
		\item \begin{math}
			F(x_1,\ldots,x_n)\text{ is a (non-commutative) homogeneous degree-}d \text{ polynomial}. 
		\end{math} \label{noncommbihompr1}
		\item $F(X_1,\ldots,X_n)\in\mathrm{U}(\mathbb{C}^n \otimes \mathcal{H})\quad\forall (X_1,\ldots,X_n)\in\bigtimes_{\ell=1}^n \mathrm{U}(\mathcal{H})\quad\forall \mathcal{H}\text{ Hilbert space}.$ \label{noncommbihompr2}
		\item $\mathrm{det}(F(X_1,\ldots,X_n))=(\mathrm{det}(X_1\cdot\ldots\cdot X_n))^d\quad\forall (X_1,\ldots,X_n)\in\bigtimes_{\ell=1}^n \mathbb{C}^{k\times k}\quad\forall k\in \mathbb{N}.$ \label{noncommbihompr3}
	\end{enumerate}
\end{theorem}
\begin{proof}
	\phantom{``$\Longrightarrow\colon$''}
	\begin{enumerate}
		\item[\ref{noncommbihompr1}] $U_i$ are constants for all $i\in\{0,\ldots, d\}$ and $\mathrm{diag}(x_1,\ldots,x_n)$ has homogeneous degree~$1$.
		\item[\ref{noncommbihompr2}] $U_i\otimes I_{\mathcal{H}}$ and $\mathrm{diag}(X_1,\ldots,X_n)$ are in $\mathrm{U}(\mathbb{C}^n \otimes \mathcal{H})\quad\forall (X_1,\ldots,X_n)\in\bigtimes_{\ell=1}^n \mathrm{U}(\mathcal{H})$.
		\item[\ref{noncommbihompr3}] Follows from the multiplicativity of determinants since $U_i\otimes I_k\in \mathrm{SU(n \cdot k)}$.	\qedhere	
	\end{enumerate}
\end{proof}

It is natural to ask whether the converse direction of \autoref{thm:noncommhombichar} holds similarly to \autoref{thm:hombichar}. 
In the bivariate case when the coefficient of $x_1^d$ or $x_2^d$ is non-zero, then the commutative polynomial corresponding to $F(x_1,x_2)$ has an essentially unique representation due to Haah's result Eq.~\eqref{eq:haah_decomp}, which in turn also describes a corresponding non-commutative polynomial $\tilde{F}(x_1,x_2)$. The question becomes whether the commutative version of a bivariate polynomial $F(x_1,x_2)$ uniquely determines its non-commuting counterpart under the above conditions. 
We leave these as open questions for further study.

\section{Problem formulation of the alternative multivariate case}
A possible alternative bivariate generalization of QSP along the lines of~\cite{rossi2022multivariableQSP} is to consider what happens when one has access to two black-box unitaries $V_i:=\ketbra{0}{0}\otimes W_i+\ketbra{1}{1}\otimes W_i^{\dagger}$ for $i\in \{1,2\}$. Assuming that $W_1$ and $W_2$ commute, we once again consider a fixed common eigenvector $\ket{\theta}$ with  corresponding eigenvalues $a=e^{i\theta}$ and $b=e^{i\varphi}$ respectively. The input signal then becomes a pair $(a,b)\in \mathrm{U(1)}\times\mathrm{U(1)}$ and the following polynomial transformations can be realized by repeated calls to the black-box unitaries modulated with fixed unitary gates on the first qubit:

\begin{equation}\label{eq:laur_bi}
	F(a,b)=U_0\begin{pmatrix}c_1&0\\0&c_1^{-1}
	\end{pmatrix}U_1\cdot\ldots\cdot\begin{pmatrix}
	c_d&0\\0&c_d^{-1}
	\end{pmatrix}U_d\text{, where }U_i\in\mathrm{SU(2)},c_i\in\{a,b\}.
\end{equation}

$F(a,b)$ is now a bivariate Laurent polynomial in the variables $a,b$ satisfying the properties summarized in the theorem below, providing a common generalization to the results of~\cite{haah2018ProdDecPerFuncQSignPRoc,rossi2022multivariableQSP}:
\begin{theorem}[Necessary conditions for product decomposition]\label{binec}
	Suppose $F(a,b)$ is a matrix-valued bivariate Laurent polynomial that can be written as a product \eqref{eq:laur_bi}. Let $(d_a,d_b)\in\mathbb{N}\times\mathbb{N}$ be the number of times $c_i=a$ and $c_i=b$ respectively, then $F(a,b)$ satisfies:
	
	\begin{enumerate}
		\item \begin{math}
			\deg_a F(a,b)\leq d_a,\quad\deg_b F(a,b)\leq d_b
		\end{math}
		\item $F(a,b)\in\mathrm{SU(2)}\quad\forall (a,b)\in\mathrm{U(1)}\times\mathrm{U(1)}$\label{bipr2}
		\item $F(-a,b)=(-1)^{d_a}F(a,b),\quad F(a,-b)=(-1)^{d_b}F(a,b)$\label{bipr3}
	\end{enumerate}
\end{theorem}

It is a natural questions to ask, by analogy with the univariate case and \autoref{thm:hombichar}, if the conditions in \autoref{binec} are also sufficient for the existence of a product decomposition \eqref{eq:laur_bi}. The answer to this question is positive if the degree is no more than $1$ in one of the two variables, but negative in general, as illustrated by the theorem below. This also refutes the main conjecture of~\cite[Conjecture~2.1]{rossi2022multivariableQSP}, since the presented counterexample possesses the additional symmetries required, as described by \cite{mori2023}.

\begin{theorem}\label{main}
	Suppose $F(a,b)$ is a bivariate Laurent polynomial in $a,b$ satisfying the properties listed in \autoref{binec}. If $\deg_a F(a,b)\leq 1$ or $\deg_b F(a,b)\leq 1$ then $F(a,b)$ can always be written as a product \eqref{eq:laur_bi}, but not necessarily if both $\deg_a F(a,b)\geq 2$ and $\deg_b F(a,b)\geq 2$. An explicit counterexample is given by:
	\begin{equation}
		F_{2,2}(a,b)=\begin{pmatrix}
			P_{2,2}(a,b)&Q_{2,2}(a,b)\\
			-Q_{2,2}^*\left(a^{-1},b^{-1}\right)&P_{2,2}^*\left(a^{-1},b^{-1}\right)
		\end{pmatrix},
	\end{equation}
	where
	\begin{multline}
		P_{2,2}(a,b)=\frac{6}{25}\sqrt{\frac{37}{493}}\Bigg
		[a^2b^2+a^{-2}b^{-2}-\left(\frac{122}{37}+\frac{8i}{37}\right)\left(b^2+b^{-2}\right)+\\ +\left(\frac{114}{37}+\frac{56i}{37}\right)\left(a^{-2}b^{2}+a^2b^{-2}\right)
		+\left(\frac{362}{111}-\frac{248i}{111}\right)\left(a^2+a^{-2}\right)+\frac{692}{111}-\frac{719i}{222}\Bigg],
	\end{multline}
	\begin{multline}
		Q_{2,2}(a,b)=\frac{6}{25}\sqrt{\frac{37}{493}}\Bigg
		[a^2b^2-a^{-2}b^{-2}-\left(\frac{122}{37}+\frac{66i}{37}\right)\left(b^2-b^{-2}\right)+\\ +\left(\frac{56}{37}+\frac{114i}{37}\right)\left(a^{-2}b^2-a^2b^{-2}\right)
		+\left(\frac{362}{111}-\frac{418i}{111}\right)\left(a^2-a^{-2}\right)\Bigg].
	\end{multline}
\end{theorem}
\begin{proof}
	First we show that if $\deg_a F(a,b)\leq 1$ or $\deg_b F(a,b)\leq 1$ then $F(a,b)$ can always be written as a product \eqref{eq:laur_bi}. If $\deg_a F(a,b)=0$ or $\deg_b F(a,b)=0$ then the result follows immediately from the univariate case. Assume therefore without loss of generality that $\deg_a F(a,b)=1$, let $d=\deg_b F(a,b)$, and write
	\begin{equation}
		F(a,b)=M_1(b)a+M_{-1}(b)a^{-1}
	\end{equation}
	for some matrix-valued univariate Laurent polynomials $M_1(b)$ and $M_{-1}(b)$. Note also that for fixed $b\in\mathrm{U(1)}$ we can view $F(a,b)$ as a univariate matrix-valued Laurent polynomial in $a$. Using the Haah decomposition theorem we can write:
	\begin{equation}\label{eq:deg1dec}
		F(a,b)=E_0(b)E_{\Pi(b)}(a)
	\end{equation}
	for $E_0(b)\in\mathrm{SU(2)}$ and $\Pi(b)$ a rank-one projection matrix, both of which are functions of $b$. Observe that:
	\begin{equation}
		E_0(b)=F(1,b),\quad \Pi(b)=\left[E_0(b)\right]^\dagger M_1(b)=\left[F(1,b)\right]^\dagger M_1(b),
	\end{equation}
	hence both $E_0(b)$ and $\Pi(b)$ are univariate matrix-valued Laurent polynomials in the variable $b$. Note also that:
	\begin{equation}
		\Pi(-b)=\left[F(1,-b)\right]^\dagger M_1(-b)=(-1)^d\left[F(1,b)\right]^\dagger (-1)^dM_1(b)=\Pi(b),
	\end{equation}
	therefore $\Pi(b)$ is an even polynomial. As the following proposition shows, such polynomials are quite special:
	
	\begin{prop}
		Suppose $\Pi(b)\in\mathrm{Mat}(2;\mathbb{C})\left[b,b^{-1}\right]$ is a degree $2k$ Laurent polynomial such that $\Pi(b)$ is a rank-one projection matrix for any $b\in\mathrm{U(1)}$. There exists a Laurent polynomial $U(b)\in\mathrm{Mat}(2;\mathbb{C})\left[b,b^{-1}\right]$ of degree $k$ and a non-zero $\ket\psi\in\mathbb{C}^2$ such that $U(b)\in\mathrm{SU(2)}$ and $\Pi(b)=U(b)\ketbra{\psi}{\psi}\left[U(b)\right]^\dagger$ for $\forall b\in\mathrm{U(1)}$. 
	\end{prop}
	\begin{proof}
		We will prove the proposition by induction on $k$. The base case $k=0$ is trivial because $\Pi(b)$ is then constant. Suppose therefore that we have proved the statement for polynomials of degree smaller than $2k\geq 2$. Let
		\begin{equation}
			\Pi(b)=\sum_{j=-2k}^{2k}\Pi_jb^j
		\end{equation}
				be a polynomial satisfying the properties listed in the proposition. As $\Pi(b)$ is a projection matrix for all $b\in\mathrm{U(1)}$, we have:
		\begin{equation}
			\Pi(b)=\left[\Pi(b)\right]^\dagger=\Pi^\dagger\left(b^{-1}\right). 
		\end{equation}
		This implies that $\Pi_j=\Pi^{\dagger}_{-j}$ for all even integers $j$ between $-2k$ and $2k$. Furthermore,
		\begin{equation}\label{eq:proj_sq}
			\Pi(b)\Pi(b)=\Pi(b)
		\end{equation}
		for all $b\in\mathrm U(1)$. Comparing the leading coefficients on both sides we find that $\Pi_{2k}^2=\Pi_{-2k}^2=0$. As the degree of $\Pi(b)$ is $2k$, this implies $\Pi_{2k}\neq 0$ and $\Pi_{-2k}=\Pi_{2k}^\dagger\neq 0$, so we deduce that $\Pi_{2k}$ is a $2\times 2$ matrix of rank one with complex entries. As such, it can be written in the form $\Pi_{2k}=\ketbra{\psi}{\phi}$ 
		for some non-zero vectors $\ket{\psi}$, $\ket\phi\in\mathbb{C}^2$. In addition, since $\Pi_{2k}^2=0$ these vectors are also orthogonal to each other. Define the following projection matrices:
		\begin{equation}
			S=\frac{\ketbra{\phi}{\phi}}{\braket{\phi}{\phi}},\quad T=I-S=\frac{\ketbra{\psi}{\psi}}{\braket{\psi}{\psi}}.
		\end{equation}
		Let \begin{equation}\label{eq:proj_decr}\tilde{\Pi}(b)=E_S(b)\Pi(b)\left[E_S(b)\right]^\dagger.\end{equation}
		We will show that $\tilde{\Pi}(b)$ has strictly smaller degree than $2k$ and it satisfies all the properties listed in the proposition. It is clear that $\tilde{\Pi}(b)$ is a matrix-valued Laurent polynomial in $b$. It is a projection for $b\in\mathrm{U(1)}$ because it is obtained by conjugating a projection matrix $\Pi(b)$ by a unitary matrix $E_S(b)$. It is also an even polynomial because it is a product of two odd  polynomials $E_S(b)$ and $\left[E_S(b)\right]^\dagger$ and an even polynomial $\Pi(b)$.
		
		Let us now expand the right hand side of \eqref{eq:proj_decr}:
		\begin{align}
            \begin{split}
                \tilde{\Pi}(b)&=\left(bS+b^{-1}T\right)\left(\Pi_{2k}b^{2k}+\ldots+\Pi_{-2k}b^{-2k}\right)\left(bT+b^{-1}S\right)\\
            &=(S\Pi_{2k}T)b^{2k+2}+(S\Pi_{2k}S+T\Pi_{2k}T+S\Pi_{2k-2}T)b^{2k}+\ldots\\
            &\ldots+(T\Pi_{-2k}T+S\Pi_{-2k}S+T\Pi_{-(2k-2)}S)b^{-2k}+T\Pi_{-2k}Sb^{-(2k+2)}
            \end{split}
		\end{align}
		
		Note that $S\Pi_{2k}=\Pi_{2k}T=\Pi_{-2k}S=T\Pi_{2k}=0$ due to the orthogonality condition $\braket{\psi}{\phi}=0$. Therefore in order to prove that $\deg \tilde{\Pi}(b)<2k$ it suffices to show that $S\Pi_{2k-2}T=T\Pi_{-(2k-2)}S=0$. Comparing coefficients of $b^{4k-2}$ in \eqref{eq:proj_sq} we find:
		\begin{equation}\label{eq:proj_lead_2}
			\Pi_{2k}\Pi_{2k-2}+\Pi_{2k-2}\Pi_{2k}=\begin{cases}
				0&\text{ if }k > 1\\
				\Pi_{2k}&\text{ if }k = 1
			\end{cases}.
		\end{equation}
		Combining \eqref{eq:proj_lead_2} and $\Pi_{2k}^2=0$ we find:
		\begin{equation}\label{eq:proj_ort}
			0=\Pi_{2k}\Pi_{2k-2}\Pi_{2k}=\ketbra{\psi}{\phi}\Pi_{2k-2}\ketbra{\psi}{\phi}=\bra{\phi}\Pi_{2k-2}\ket{\psi}\Pi_{2k}
		\end{equation}
		Hence $\bra{\phi}\Pi_{2k-2}\ket{\psi}=0$, thus $S\Pi_{2k-2}T=\ketbra{\phi}{\phi}\Pi_{2k-2}\ketbra{\psi}{\psi}=0$. Finally, $T\Pi_{-(2k-2)}S=\left(S\Pi_{2k-2}T\right)^\dagger=0$, therefore $\deg\tilde{\Pi}(b)<2k$. Using the induction hypothesis we may express $\tilde{\Pi}(b)$ as a product and insert it in \eqref{eq:proj_decr} to obtain a product decomposition of $\Pi(b)$ which concludes the proof of the proposition.
	\end{proof}	
	Using the proposition, we can write $\Pi(b)=U(b)K\left[U(b)\right]^\dagger$ in \eqref{eq:deg1dec} for a constant rank-one projection matrix $K=\ketbra{\psi}{\psi}$ to find:
	\begin{equation}
		F(a,b)=E_0(b)E_{U(b)K\left[U(b)\right]^\dagger}(a)=E_0(b)U(b)E_K(a)\left[U(b)\right]^\dagger,
	\end{equation}	
	which is equivalent to a product decomposition of the form \eqref{eq:univ_laur}, which proves the first part of the theorem.
	
	It follows from the above that in order to find a counterexample, it is necessary to have $\deg_a F(a,b)\geq 2$ and $\deg_b F(a,b)\geq 2$. From now on, we will concentrate on the case when $\deg_a F(a,b)=\deg_b F(a,b)=2$ and look for a counterexample in the form	\begin{equation}\label{eq:Fentries}
		F(a,b)=\begin{pmatrix}
			P(a,b)&Q(a,b)\\
			-Q^*\left(a^{-1},b^{-1}\right)&P^*\left(a^{-1},b^{-1}\right)
		\end{pmatrix},
	\end{equation}
	where $P(a,b)$ and $Q(a,b)$ are scalar-valued Laurent polynomials of degree $2$ in both variables $a$ and $b$ satisfying the additional symmetries
	\begin{equation}\label{eq:poly_sym}
		P\left(a^{-1},b^{-1}\right)=P(a,b),\quad Q\left(a^{-1},b^{-1}\right)=-Q(a,b).
	\end{equation}
	The form of $F(a,b)$ in \eqref{eq:Fentries} comes from the fact that for any $U\in\mathrm{SU}(2)$ we have $U_{11}=U^*_{22}$ and $U_{12}=-U^*_{21}$, while the symmetries \eqref{eq:poly_sym} are motivated by the fact that the entries of any $F(a,b)$ that can be written as a product of the form \eqref{eq:laur_bi} when all the $U_i$-s being Pauli-Z rotations obey these. In particular, such a pair of polynomials provide an explicit counterexample to ~\cite[Conjecture~2.1]{rossi2022multivariableQSP}. Furthermore, they reduce the number of independent coefficients in $F(a,b)$. Combining the parity constraint on $F(a,b)$ and \eqref{eq:poly_sym} we may expand: 
	\begin{align}
		P(a,b)&=\left(Aa^2+B+Ca^{-2}\right)b^2+Da^2+E+Da^{-2}+\left(Aa^{-2}+B+Ca^2\right)b^{-2}\\
		Q(a,b)&=\left(Fa^2+G+Ha^{-2}\right)b^2+I\left(a^2-a^{-2}\right)-\left(Fa^{-2}+G+Ha^2\right)b^{-2},
	\end{align}
	where $A,B,\ldots, H\in\mathbb{C}$. As before, we can think of $P(a,b)$ and $Q(a,b)$ as univariate polynomials in $a$ or $b$ with univariate polynomial coefficients in $b$ or $a$ by writing:
	\begin{align}
		P(a,b)&=P_2^a\left(a^2\right)b^2+P_0^a\left(a^2\right)+P_{-2}^a\left(a^2\right)b^{-2}=P_2^b\left(b^2\right)a^2+P_0^b\left(b^2\right)+P_{-2}^b\left(b^2\right)a^2\\
		Q(a,b)&=Q_2^a\left(a^2\right)b^2+Q_0^a\left(a^2\right)+Q_{-2}^a\left(a^2\right)b^{-2}=Q_2^b\left(b^2\right)a^2+Q_0^b\left(b^2\right)+Q_{-2}^b\left(b^2\right)a^{-2},
	\end{align}
	where the univariate polynomials $P_2^a(t),\dots,Q_{-2}^b(t)$ are defined as
	\begin{align}
		P_2^a(t)=P_{-2}^a\left(t^{-1}\right)&=At+B+Ct^{-1}=At^{-1}\left(t-\alpha_0\right)\left(t-\alpha_1\right)\\
		Q_2^a(t)=Q_{-2}^a\left(t^{-1}\right)&=Ft+G+Ht^{-1}=Ft^{-1}\left(t-\beta_0\right)\left(t-\beta_1\right)\\
		P_0^a(t)&=D\left(t+t^{-1}\right)+E\\
		Q_0^a(t)&=I\left(t-t^{-1}\right)\\
		P_2^b(t)=P_{-2}^b\left(t^{-1}\right)&=At+D+Ct^{-1}=At^{-1}\left(t-\gamma_0\right)\left(t-\gamma_1\right)\\
		Q_2^b(t)=Q_{-2}^b\left(t^{-1}\right)&=Ft+I-Ht^{-1}=Ft^{-1}\left(t-\delta_0\right)\left(t-\delta_1\right)\\
		P_0^b(t)&=B\left(t+t^{-1}\right)+E\\
		Q_0^b(t)&=G\left(t-t^{-1}\right),
	\end{align}
	with $\alpha_0,\alpha_1,\beta_0,\beta_1,\gamma_0,\gamma_1,\delta_0,\delta_1\in\mathbb{C}$ being the roots of $P_2^a(t),Q_2^a(t),P_2^b(t)$ and $Q_2^b(t)$, respectively.
	Property \ref{bipr2} in \autoref{binec} is captured by the following polynomial equation:
	\begin{equation}\label{eq:matpoly_unit_eq}
		F(a,b)\left[F(a,b)\right]^\dagger=F(a,b)F^\dagger\left(a^{-1},b^{-1}\right)=I
	\end{equation}
	In terms of the polynomials $P(a,b)$ and $Q(a,b)$ \eqref{eq:matpoly_unit_eq} reads, using \eqref{eq:poly_sym}:
	\begin{equation}\label{eq:pol_unit_sym}
		|P(a,b)|^2+|Q(a,b)|^2=P(a,b)P^*(a,b)-Q(a,b)Q^*(a,b)=1.
	\end{equation}
	By looking at the coefficient of $a^4b^4$, we find that $|A|^2=|F|^2$. Comparing the coefficients of $a^4$ and $b^4$ on both sides of \eqref{eq:pol_unit_sym} we find:
	\begin{align}
		P_2^a(t)\left(P_2^a\right)^*(t)&=Q_2^a(t)\left(Q_2^a\right)^*(t)\label{eq:abs_eqa}\\
		P_2^b(t)\left(P_2^b\right)^*(t)&=Q_2^b(t)\left(Q_2^b\right)^*(t)\label{eq:abs_eqb}.
	\end{align}
	Equation \eqref{eq:abs_eqa} implies that the pairs of roots $\{\alpha_0,\alpha_1\}$ and $\{\beta_0,\beta_1\}$ are equal up to complex conjugation and analogously for $\{\gamma_0,\gamma_1\}$ and $\{\delta_0,\delta_1\}$. This is a point in the proof where a crucial difference shows up between the univariate and bivariate case: for univariate $F(t)$, equations \eqref{eq:abs_eqa} and \eqref{eq:abs_eqb} relate complex numbers (leading coefficients of entries of $F(t)$) rather than polynomials, and while the equality of the magnitudes of two complex numbers implies that they differ only by phase, this is no longer true for polynomials.
	Let us look for a counterexample by setting
	\begin{equation}
		\beta_0=\alpha_0,\beta_1=\alpha_1^*,\quad\delta_0=\gamma_0,\delta_1=\gamma_1^*
	\end{equation}
	Observe that we can express all coefficients $A,\ldots,H$ with the exception of $E$ in terms of $A,F,\alpha_0,\alpha_1,\gamma_0,\gamma_1$ as follows:
	\begin{align}
		B&=-A(\alpha_0+\alpha_1)\\
		C&=A\alpha_0\alpha_1=A\gamma_0\gamma_1\label{eq:cexpr}\\
		D&=-A(\gamma_0+\gamma_1)\\
		G&=-F(\alpha_0+\alpha_1^*)\\
		H&=F\alpha_0\alpha_1^*=-F\gamma_0\gamma_1^*\label{eq:hexpr}\\
		I&=-F(\gamma_0+\gamma_1^*)		
	\end{align}
	From \eqref{eq:cexpr} and \eqref{eq:hexpr} we deduce that:
	\begin{align}
		\alpha_0\alpha_1&=\gamma_0\gamma_1\\
		\alpha_0\alpha_1^*&=-\gamma_0\gamma_1^*.
	\end{align}
	This implies (assuming $\alpha_0,\alpha_1,\gamma_0,\gamma_1$ are non-zero):
	\begin{equation}
		(\gamma_1/\alpha_1)^*+\gamma_1/\alpha_1=(\gamma_0/\alpha_0)^*+\gamma_0/\alpha_0=0,
	\end{equation}
	which means that
	\begin{align}
		\gamma_0&=ik\alpha_0\\
		\gamma_1&=\alpha_1/(ik)
	\end{align}
	for some non-zero real number $k\in\mathbb{R}$. Looking at the coefficients of the monomials $a^2b^2$, $a^2b^{-2}$, $a^2$ and $b^2$ on both sides of \eqref{eq:pol_unit_sym} we learn:
	\begin{align}
		0&=AE^*+BD^*+DB^*+EA^*-GI^*-IG^*\label{eq:firsteq}\\
		0&=CE^*+C^*E+BD^*+DB^*+IG^*+I^*G\\
		0&=AB^*+BC^*+DE^*+ED^*+BA^*+CB^*+FG^*+GH^*+GF^*+HG^*\\
		0&=AD^*+BE^*+CD^*+DA^*+EB^*+DC^*+FI^*-HI^*+F^*I-H^*I.\label{eq:lasteq}
	\end{align}
	Let:
	\begin{align}
		r&:=GI^*+IG^*-BD^*-DB^*\\
		s&:=-(GI^*+IG^*+BD^*+DB^*)\\
		t&:=-\left(AB^*+BA^*+BC^*+CB^*+FG^*+GF^*+GH^*+HG^*\right)\\
		u&:=IH^*+HI^*-FI^*-IF^*-AD^*-DA^*-CD^*-DC^*.
	\end{align}
	Our systems of equations \eqref{eq:firsteq}-\eqref{eq:lasteq} for $E$ can then be written more compactly as:
	\begin{align}
		AE^*+EA^*=r\label{eq:firstlineeq}\\
		CE^*+EC^*=s\\
		DE^*+ED^*=t\\
		BE^*+EB^*=u.\label{eq:lastlineeq}
	\end{align}
	Equations \eqref{eq:firstlineeq}-\eqref{eq:lastlineeq} have a simple geometric interpretation, namely that $E$ lies at the intersection of four straight lines on the complex plane. These four lines intersect at the same point if:
	\begin{equation}\label{eq:final_eq}
		\frac{rC-sA}{CA^*-AC^*}=\frac{rD-tA}{DA^*-AD^*}=\frac{rB-uA}{BA^*-AB^*}.
	\end{equation}
	In order to find a solution, we need to choose $\alpha_0$, $\alpha_1$ and $k$ such that \eqref{eq:final_eq} holds. For fixed $\alpha_0\in\mathbb{C}$ and $k\in\mathbb{R}$ equation \eqref{eq:final_eq} constitutes a system of polynomial equations for the unknown real and imaginary parts of $\alpha_1$. The following exact solution to this system of equations was obtained using Wolfram Mathematica with the choices $\alpha_0=1+i,k=3$:
	\begin{align}
		\alpha_1&=\frac{85}{37}-\frac{29i}{37}\label{eq:alph}\\
		\frac{E}{A}&=\frac{692}{111}-\frac{719i}{222}\label{eq:const_ratio}.
	\end{align}
	To specify the polynomials completely, the only thing remaining is to fix the constant term which determines $|A|$. (Observe that the phases of $A$ and $F$ are arbitrary, because if $P(a,b),Q(a,b)$ is a solution, so is $e^{i\varphi}P(a,b), e^{i\vartheta}Q(a,b)$ for arbitrary angles $\varphi$ and $\vartheta$.) From \eqref{eq:pol_unit_sym} we have:
	\begin{equation}
		2|A|^2+2|B|^2+2|C|^2+2|D|^2+|E|^2+2|F|^2+2|G|^2+2|H|^2+2|I|^2=1.
	\end{equation}
	Substituting the above results we find that:
	\begin{equation}
		|A|=\frac{6}{25}\sqrt{\frac{37}{493}}
	\end{equation}
	Finally, the full form of the polynomials $P(a,b),Q(a,b)$ is (for arbitrary angles $\varphi,\vartheta$ as discussed above):
	\begin{multline}
		P(a,b)=\frac{6e^{i\varphi}}{25}\sqrt{\frac{37}{493}}\Bigg
		[a^2b^2+a^{-2}b^{-2}-\left(\frac{122}{37}+\frac{8i}{37}\right)\left(b^2+b^{-2}\right)+\\ \left(\frac{114}{37}+\frac{56i}{37}\right)\left(a^{-2}b^{2}+a^2b^{-2}\right)
		+\left(\frac{362}{111}-\frac{248i}{111}\right)\left(a^2+a^{-2}\right)+\frac{692}{111}-\frac{719i}{222}\Bigg]
	\end{multline}
	\begin{multline}
		Q(a,b)=\frac{6e^{i\vartheta}}{25}\sqrt{\frac{37}{493}}\Bigg
		[a^2b^2-a^{-2}b^{-2}-\left(\frac{122}{37}+\frac{66i}{37}\right)\left(b^2-b^{-2}\right)+\\+\left(\frac{56}{37}+\frac{114i}{37}\right)\left(a^{-2}b^2-a^2b^{-2}\right)
		+\left(\frac{362}{111}-\frac{418i}{111}\right)\left(a^2-a^{-2}\right)\Bigg]
	\end{multline}
	In order to show without doubt that this is a valid counterexample, we give a necessary condition on the coefficients of a decomposable polynomial \eqref{eq:laur_bi} and show that our example does not satisfy this condition.
	\begin{lemma}\label{lem:furth_nec}
		Let $F(a,b)$ be a matrix-valued bivariate Laurent polynomial that can be decomposed as a product \eqref{eq:laur_bi}. Write
		\begin{equation}
			F(a,b)=\sum_{j=-d_a}^{d_a}\sum_{k=-d_b}^{d_b}M_{j,k}a^jb^k,
		\end{equation}
		where $d_a=\deg_a F(a,b), d_b=\deg_b  F(a,b)$. Then:
		\begin{equation}
			M_{d_a,d_b}M_{d_a,-d_b}^\dagger=0\text{ or }M_{d_a,-d_b}M_{-d_a,-d_b}^\dagger=0.
		\end{equation}
	\end{lemma}
	\begin{proof}
		As a result of Property \ref{bipr2} in \autoref{binec} $F(a,b)$ satisfies the following: 
		\begin{align}
			F(a,b)F^\dagger\left(a^{-1},b^{-1}\right)&=I\\
			F^\dagger\left(a^{-1},b^{-1}\right)F(a,b)&=I\\
			\det F(a,b)&=1
		\end{align}
		Comparing the leading coefficients on both sides of these equations we learn that
		\begin{equation}
			M_{d_a,d_b}^\dagger M_{-d_a,-d_b}=M_{-d_a,-d_b}M_{d_a,d_b}^\dagger=M_{d_a,-d_b}^\dagger M_{-d_a,d_b}=M_{-d_a,d_b}M_{d_a,-d_b}^\dagger=0
		\end{equation}
		and also that:
		\begin{equation}\label{eq:det_eq}
			\det M_{\pm d_a,\pm d_b}=0.
		\end{equation}
		Assuming that none of $M_{\pm d_a,\pm d_b}$ is zero, equation \eqref{eq:det_eq} implies that all have rank one, thus we may write:
		\begin{align}
			M_{d_a,d_b}&=\alpha\ketbra{x}{z}\\
			M_{-d_a,-d_b}&=\beta\ketbra{y}{w}\\
			M_{d_a,-d_b}&=\gamma\ketbra{p}{r}\\
			M_{-d_a,d_b}&=\delta\ketbra{q}{s},
		\end{align}
		where $\alpha,\dots,\delta\in\mathbb{C}$ are complex numbers and $\ket x,\dots,\ket s\in\mathbb{C}^2$ are two-dimensional complex unit vectors satisfying
		\begin{equation}
			\braket{x}{y}=\braket{z}{w}=\braket{p}{q}=\braket{r}{s}=0.
		\end{equation}
		If $F(a,b)$ can be decomposed as a product \eqref{eq:laur_bi}, then there exists a matrix $U\in\mathrm{SU}(2)$ and a choice of variable $c\in\{a,b\}$ such that the polynomial
		\begin{equation}\label{eq:deg_decr}
			F(a,b)U^\dagger\begin{pmatrix}
				c^{-1}&0\\
				0&c
			\end{pmatrix}
		\end{equation}
		has smaller degree in the variable $c$ than $F(a,b)$. Suppose that $c=a$ and expanding \eqref{eq:deg_decr}:
		\begin{equation}\label{eq:deg_decr_exp}
			\begin{split}
				F(a,b)U^\dagger\begin{pmatrix}
					a^{-1}&0\\
					0&a
				\end{pmatrix}&=\left(M_{\pm d_a,\pm d_b}a^{\pm d_a}b^{\pm d_b}+\dots\right)U^\dagger\left(a^{-1}\ketbra{0}{0}+a\ketbra{1}{1}\right)=\\
				&=M_{d_a,d_b}U^\dagger\ketbra{1}{1}a^{d_a+1}b^{d_b}+M_{-d_a,-d_b}U^\dagger\ketbra{0}{0}a^{-d_a-1}b^{-d_b}+\\
				&+M_{d_a,-d_b}U^\dagger\ketbra{1}{1}a^{d_a+1}b^{-d_b}+M_{-d_a,d_b}U^\dagger\ketbra{0}{0}a^{-d_a-1}b^{d_b}+\dots
			\end{split}
		\end{equation}
		In order for the degree in $a$ to decrease in \eqref{eq:deg_decr_exp}, all four terms on the RHS have to vanish, yielding:
		\begin{equation}
			M_{d_a,d_b}U^\dagger\ket{1}=M_{-d_a,-d_b}U^\dagger\ket{0}=M_{d_a,-d_b}U^\dagger\ket{1}=M_{-d_a,d_b}U^\dagger\ket{0}=0.
		\end{equation}
		This means that:
		\begin{equation}
			\mel{1}{U}{z}=\mel{0}{U}{w}=\mel{1}{U}{r}=\mel{0}{U}{s}=0.
		\end{equation}
		Therefore:
		\begin{equation}
			U=e^{i\varphi}\left(\ketbra{0}{z}+\ketbra{1}{w}\right)=e^{i\vartheta}\left(\ketbra{0}{r}+\ketbra{1}{s}\right)
		\end{equation}
		for some angles $\varphi, \vartheta$. This implies that $\{\ket z,\ket w\}$ are parallel with $\{\ket r,\ket s\}$, respectively, hence $\braket{z}{s}=\braket{w}{r}=0$, thus $M_{d_a,d_b}M_{-d_a,d_b}^\dagger=M_{-d_a,-d_b}M_{d_a,-d_b}^\dagger=0$. Similarly, if $c=b$, then we find that $M_{d_a,d_b}M_{d_a,-d_b}^\dagger=M_{-d_a,-d_b}M_{-d_a,d_b}^\dagger=0$.
	\end{proof}
	Returning to our counterexample, we have:
	\begin{align}
		M_{2,2}&=\begin{bmatrix}
			A&F\\
			F^*&A^*
		\end{bmatrix}\\
		M_{2,-2}&=\begin{bmatrix}
			C&-H\\
			-H^*&C^*
		\end{bmatrix}\\
		M_{-2,-2}&=\begin{bmatrix}
			A&-F\\
			-F^*&A^*
		\end{bmatrix}
	\end{align}
	Computing the matrix products below:
	\begin{equation}
		\begin{split}
			M_{2,2}M_{2,-2}^\dagger&=\begin{bmatrix}
				A&F\\
				F^*&A^*
			\end{bmatrix}\begin{bmatrix}
				C^*&-H\\
				-H^*&C
			\end{bmatrix}=\begin{bmatrix}
				AC^*-FH^*&-AH+CF\\
				C^*F^*-A^*H^*&A^*C-F^*H
			\end{bmatrix}=\\
			&=\begin{bmatrix}
				\frac{72}{10625}(1+i)&.\\
				.&\frac{72}{10625}(1-i)
			\end{bmatrix}\neq 0,
		\end{split}
	\end{equation}
	\begin{equation}
		\begin{split}
			M_{2,-2}M_{-2,-2}^\dagger&=\begin{bmatrix}
				C&-H\\
				-H^*&C^*
			\end{bmatrix}\begin{bmatrix}
				A^*&-F\\
				-F^*&A
			\end{bmatrix}=\begin{bmatrix}
				A^*C+F^*H&-FC-AH\\
				-A^*H^*-F^*C^*&AC^*+FH^*
			\end{bmatrix}=\\
			&=\begin{bmatrix}
				\frac{72}{3625}(1+i)&.\\
				.&\frac{72}{3625}(1-i)
			\end{bmatrix}\neq 0,
		\end{split}
	\end{equation}
	by \autoref{lem:furth_nec} we conclude that the polynomial $F_{2,2}(a,b)$ given in \autoref{main} cannot be decomposed as a product \eqref{eq:laur_bi}.
\end{proof}

\section{Conclusion and open questions}
We have introduced new general homogeneous variants of multivariate quantum signal processing, and presented partial results in their characterization. In the bivariate commuting case we derived a complete characterization by drawing a connection to the univariate case, in turn showing an ``equivalence'' between Haah's general QSP~\cite{haah2018ProdDecPerFuncQSignPRoc} and a recent QSP variant introduced by Motlagh and Wiebe~\cite{motlagh2023GeneralQSP}. This connection in turn shows that the very efficient algorithm of~\cite{motlagh2023GeneralQSP} for finding the $\mathrm{SU}(2)$ modulation gate sequences can be readily applied to other variants of QSP as well.

Our MQSP variants break away from the earlier two-dimensional treatment limited by its reliance on Jordan-like decompositions, and thus might ultimately lead to the development of novel quantum algorithms. While studying non-commuting extension of the alternative MQSP variants considered by Rossi and Chuang~\cite{rossi2022multivariableQSP} could also lead to new characterizations and insights, the combinatorial freedom in the places of the variable insertions might make their algebraic treatment more difficult. 

A major open question is whether the necessary conditions for our homogeneous MQSP
\begin{enumerate}[label=\roman*)]
	\item \begin{math}
		F(x_1,\ldots,x_n)\text{ is a (non-commutative) homogeneous degree-}d \text{ polynomial},
	\end{math}
	\item $F(X_1,\ldots,X_n)\in\mathrm{U}(\mathbb{C}^n \otimes \mathcal{H})\quad\forall (X_1,\ldots,X_n)\in\bigtimes_{\ell=1}^n \mathrm{U}(\mathcal{H})\quad\forall \mathcal{H}\text{ Hilbert space},$ 
	\item $\mathrm{det}(F(X_1,\ldots,X_n))=(\mathrm{det}(X_1\cdot\ldots\cdot X_n))^d\quad\forall (X_1,\ldots,X_n)\in\bigtimes_{\ell=1}^n \mathbb{C}^{k\times k}\quad\forall k\in \mathbb{N},$
\end{enumerate}
are also sufficient for guaranteeing a product decomposition of $F(x_1,\ldots,x_n)$. Understanding this question in the $d\geq 3$ commuting and $d\geq 2$ non-commuting case, together with characterizations of particular matrix elements of $F(x_1,\ldots,x_n)$ analogous to~\cite[Lemma 4]{haah2018ProdDecPerFuncQSignPRoc} and \cite[Theorem 4]{motlagh2023GeneralQSP} would be a major step towards finding novel applications of MQSP. 

On top of analytical proofs, we developed a Python program package (see \autoref{apx:program}) to help finding counterexamples, which we also made publicly available~\cite{gitHubRepo2023}. This program and its further extensions might help finding potential counterexamples of the proposed various characterizations, if they exist.

We hope that our new MQSP variants and results about their characterization will guide further research and help finding
useful characterisations of “achievable” MQSP polynomials ultimately leading to the development of more efficient quantum algorithms.

\section*{Acknowledgements}
We are grateful to Zane Rossi for useful discussions. Part of this research project was conducted during the \href{https://coge.elte.hu/reu22.html}{REU 2022 program}, we are grateful for the support of its organisers. We also acknowledge funding by the EU's Horizon 2020 Marie Skłodowska-Curie program QuantOrder-891889 and the QuantERA II project \href{https://quantera.eu/hqcc/}{HQCC}-101017733 in coordination with the national funding organisation NKFIH.

\newpage

\bibliographystyle{alphaUrlePrint.bst}
\bibliography{Bibliography,LocalBibliography}

\newcommand{\etalchar}[1]{$^{#1}$}
\newcommand{\lName}{1}\newcommand{\arxiv}[1]{arXiv:
  \href{https://arxiv.org/abs/#1}{\ttfamily{#1}}\removefirstdot}\newcommand{\arXiv}[1]{arXiv:
  \href{https://arxiv.org/abs/#1}{\ttfamily{#1}}\removefirstdot}\def\removefirstdot#1{\if.#1{}\else#1\fi}\providecommand{\multiletter}[1]{#1}\renewcommand{\multiletter}[1]{#1}\DeclareRobustCommand{\dutchPrefix}[2]{#2}\providecommand{\dutchPrefix}[2]{#2}\renewcommand{\dutchPrefix}[2]{#2}\newcommand{\skp}[3]{#2}\newcommand{\focs
  }[1]{\if\lName1\skp{ }{Proceedings of the #1 {IEEE} Symposium on Foundations
  of Computer Science ({FOCS})}{ }\else{FOCS}\fi}\newcommand{\stoc
  }[1]{\if\lName1\skp{ }{Proceedings of the #1 {ACM} Symposium on the Theory of
  Computing ({STOC})}{ }\else{STOC}\fi}\newcommand{\soda }[1]{\if\lName1\skp{
  }{Proceedings of the #1 {ACM-SIAM} Symposium on Discrete Algorithms
  ({SODA})}{ }\else{SODA}\fi}\newcommand{\stacs }[1]{\if\lName1\skp{
  }{Proceedings of the #1 Symposium on Theoretical Aspects of Computer Science
  ({STACS})}{ }\else{STACS}\fi}\newcommand{\itcs }[1]{\if\lName1\skp{
  }{Proceedings of the #1 Innovations in Theoretical Computer Science
  Conference ({ITCS})}{ }\else{ITCS}\fi}\newcommand{\fsttcs
  }[1]{\if\lName1\skp{ }{Proceedings of the #1 International Conference on
  Foundations of Software Technology and Theoretical Computer Science
  ({FSTTCS})}{ }\else{FSTTCS}\fi}\newcommand{\mfcs }[1]{\if\lName1\skp{
  }{Proceedings of the #1 International Symposium on Mathematical Foundations
  of Computer Science ({MFCS})}{ }\else{MFCS}\fi}\newcommand{\ccc
  }[1]{\if\lName1\skp{ }{Proceedings of the #1 {IEEE} Conference on
  Computational Complexity ({CCC})}{ }\else{CCC}\fi}\newcommand{\isit
  }[1]{\if\lName1\skp{ }{Proceedings of the #1 {IEEE} International Symposium
  on Information Theory ({ISIT})}{ }\else{ISIT}\fi}\newcommand{\colt
  }[1]{\if\lName1\skp{ }{Proceedings of the #1 Conference On Learning Theory
  ({COLT})}{ }\else{COLT}\fi}\newcommand{\nips }[1]{\if\lName1\skp{ }{Advances
  in Neural Information Processing Systems #1 ({NIPS})}{
  }\else{NIPS}\fi}\newcommand{\aistats }[1]{\if\lName1\skp{ }{Proceedings of
  the #1 International Conference on Artificial Intelligence and Statistics
  ({AISTATS})}{ }\else{AISTATS}\fi}\newcommand{\icml }[1]{\if\lName1\skp{
  }{Proceedings of the #1 International Conference on Machine Learning
  ({ICML})}{ }\else{ICML}\fi}\newcommand{\icalp }[1]{\if\lName1\skp{
  }{Proceedings of the #1 International Colloquium on Automata, Languages, and
  Programming ({ICALP})}{ }\else{ICALP}\fi}\newcommand{\esa
  }[1]{\if\lName1\skp{ }{Proceedings of the #1 Annual European Symposium on
  Algorithms ({ESA})}{ }\else{ESA}\fi}\newcommand{\tqc }[1]{\if\lName1\skp{
  }{Proceedings of the #1 Conference on the Theory of Quantum Computation,
  Communication, and Cryptography ({TQC})}{}\else{TQC}\fi}\newcommand{\isaac
  }[1]{\if\lName1\skp{ }{Proceedings of the #1 International Symposium on
  Algorithms and Computation ({ISAAC})}{ }\else{ISAAC}\fi}\newcommand{\jacm
  }{\if\lName1\skp{ }{Journal of the ACM}{ }\else{J. ACM}\fi}\newcommand{\acmta
  }{\if\lName1\skp{ }{ACM Transactions on Algorithms}{ }\else{{ACM} Tr.
  Alg}\fi}\newcommand{\acmtct }{\if\lName1\skp{ }{ACM Transactions on
  Computation Theory}{ }\else{ACM Tr. Comp. Th.}\fi}\newcommand{\acmtqc
  }{\if\lName1\skp{ }{ACM Transactions on Quantum Computing}{ }\else{ACM Tr.
  Quant. Comp.}\fi}\newcommand{\jams }{\if\lName1\skp{ }{Journal of the AMS}{
  }\else{J. AMS}\fi}\newcommand{\pams }{\if\lName1\skp{ }{Proceedings of the
  AMS}{ }\else{Proc. AMS}\fi}\newcommand{\linalgappl }{\if\lName1\skp{ }{Linear
  Algebra and its Applications}{ }\else{Lin. Alg. \&
  App.}\fi}\newcommand{\jalgo }{\if\lName1\skp{ }{Journal of Algorithms}{
  }\else{J. Alg.}\fi}\newcommand{\jcss }{\if\lName1\skp{ }{Journal of Computer
  and System Sciences}{ }\else{J. Comp. Sys. Sci.}\fi}\newcommand{\cc
  }{\if\lName1\skp{ }{Computational Complexity}{ }\else{Comp.
  Comp.}\fi}\newcommand{\algor }{\if\lName1\skp{ }{Algorithmica}{
  }\else{Alg.}\fi}\newcommand{\comb }{\if\lName1\skp{ }{Combinatorica}{
  }\else{Comb.}\fi}\newcommand{\cacm }{\if\lName1\skp{ }{Communications of the
  ACM}{ }\else{Comm. ACM}\fi}\newcommand{\sigart }{\if\lName1\skp{ }{SIGART
  Bulletin}{ }\else{SIGART Bull.}\fi}\newcommand{\sigactn }{\if\lName1\skp{
  }{SIGACT News}{ }\else{SIGACT News}\fi}\newcommand{\eatcsbul
  }{\if\lName1\skp{ }{Bulletin of the {EATCS}}{ }\else{Bull.
  {EATCS}}\fi}\newcommand{\siamrev }{\if\lName1\skp{ }{SIAM Review}{
  }\else{SIAM Rev.}\fi}\newcommand{\siamjc }{\if\lName1\skp{ }{SIAM Journal on
  Computing}{ }\else{SIAM J. Comp.}\fi}\newcommand{\siamjo }{\if\lName1\skp{
  }{SIAM Journal on Optimization}{ }\else{SIAM J. Opt.}\fi}\newcommand{\siamjdm
  }{\if\lName1\skp{ }{SIAM Journal on Discrete Mathematics}{ }\else{SIAM J.
  Disc. Math.}\fi}\newcommand{\siamjnum }{\if\lName1\skp{ }{SIAM Journal on
  Numerical Analysis}{ }\else{SIAM J. Num. Anal.}\fi}\newcommand{\siamjmathanal
  }{\if\lName1\skp{ }{SIAM Journal on Mathematical Analysis}{ }\else{SIAM J.
  Math. Anal.}\fi}\newcommand{\discmath }{\if\lName1\skp{ }{Discrete
  Mathematics}{ }\else{Disc. Math.}\fi}\newcommand{\das }{\if\lName1\skp{
  }{Discrete Applied Mathematics}{ }\else{Disc. App.
  Math.}\fi}\newcommand{\amatstat }{\if\lName1\skp{ }{Annals of Mathematical
  Statistics}{ }\else{Ann. Math. Stat.}\fi}\newcommand{\rms }{\if\lName1\skp{
  }{Russian Mathematical Surveys}{ }\else{Russ. Math.
  Surv.}\fi}\newcommand{\invmath }{\if\lName1\skp{ }{Inventiones Mathematicae}{
  }\else{Inv. Math.}\fi}\newcommand{\jnumber }{\if\lName1\skp{ }{Journal of
  Number Theory}{ }\else{J. Num. Th.}\fi}\newcommand{\tcs }{\if\lName1\skp{
  }{Theoretical Computer Science}{ }\else{Theor. Comput.
  Sci.}\fi}\newcommand{\toc }{\if\lName1\skp{ }{Theory of Computing}{
  }\else{Th. Comp.}\fi}\newcommand{\cjtcs }{\if\lName1\skp{ }{Chicago Journal
  of Theoretical Computer Science}{}\else{Chic. J. Th. Comp.
  Sci.}\fi}\newcommand{\quantum }{\if\lName1\skp{ }{{Quantum}}{
  }\else{Quant.}\fi}\newcommand{\cmp }{\if\lName1\skp{ }{Communications in
  Mathematical Physics}{ }\else{Comm. Math. Phys.}\fi}\newcommand{\jmp
  }{\if\lName1\skp{ }{Journal of Mathematical Physics}{ }\else{J. Math.
  Phys.}\fi}\newcommand{\rspa }{\if\lName1\skp{ }{Proceedings of the Royal
  Society A}{ }\else{Proc. Roy. Soc. A}\fi}\newcommand{\qic }{\if\lName1\skp{
  }{Quantum Information and Computation}{ }\else{Quant. Inf. \&
  Comp.}\fi}\newcommand{\physrev }{\if\lName1\skp{ }{Physical Review}{
  }\else{Phys. Rev.}\fi}\newcommand{\pra }{\if\lName1\skp{ }{Physical Review
  A}{ }\else{Phys. Rev. A}\fi}\newcommand{\prb }{\if\lName1\skp{ }{Physical
  Review B}{ }\else{Phys. Rev. B}\fi}\newcommand{\pre }{\if\lName1\skp{
  }{Physical Review E}{ }\else{Phys. Rev. E}\fi}\newcommand{\prr
  }{\if\lName1\skp{ }{Physical Review Research}{ }\else{Phys. Rev.
  Research}\fi}\newcommand{\prx }{\if\lName1\skp{ }{Physical Review X}{
  }\else{Phys. Rev. X}\fi}\newcommand{\prxq }{\if\lName1\skp{ }{Physical Review
  X Quantum}{ }\else{Phys. Rev. X Quant.}\fi}\newcommand{\prl }{\if\lName1\skp{
  }{Physical Review Letters}{ }\else{Phys. Rev. Lett.}\fi}\newcommand{\njp
  }{\if\lName1\skp{ }{New Journal of Physics}{ }\else{New J.
  Phys.}\fi}\newcommand{\prapp }{\if\lName1\skp{ }{Physical Review Applied}{
  }\else{Phys. Rev. Appl.}\fi}\newcommand{\physrep }{\if\lName1\skp{ }{Physics
  Reports}{ }\else{Phys. Rep.}\fi}\newcommand{\rmp }{\if\lName1\skp{ }{Reviews
  of Modern Physics}{ }\else{Rev. Mod. Phys. }\fi}\newcommand{\phystoday
  }{\if\lName1\skp{ }{Physics Today}{ }\else{Phys.
  Today}\fi}\newcommand{\physics }{\if\lName1\skp{ }{Physics}{
  }\else{Phys.}\fi}\newcommand{\nature }{\if\lName1\skp{ }{Nature}{
  }\else{Nat.}\fi}\newcommand{\natcomm }{\if\lName1\skp{ }{Nature
  Communications}{ }\else{Nat. Comm.}\fi}\newcommand{\natphys }{\if\lName1\skp{
  }{Nature Physics}{ }\else{Nat. Phys.}\fi}\newcommand{\npjqi }{\if\lName1\skp{
  }{npj Quantum Information}{ }\else{npj Quant. Inf.}\fi}\newcommand{\scirep
  }{\if\lName1\skp{ }{Scientific Reports}{ }\else{Sci.
  Rep.}\fi}\newcommand{\science }{\if\lName1\skp{ }{Science}{
  }\else{Sci.}\fi}\newcommand{\jpa }{\if\lName1\skp{ }{Journal of Physics A:
  Mathematical and Theoretical}{ }\else{J. Phys. A}\fi}\newcommand{\ijtp
  }{\if\lName1\skp{ }{International Journal of Theoretical Physics}{
  }\else{Int. J. Th. Phys.}\fi}\newcommand{\jmo }{\if\lName1\skp{ }{Journal of
  Modern Optics}{ }\else{J. Mod. Opt.}\fi}\newcommand{\jstatph
  }{\if\lName1\skp{ }{Journal of Statistical Physics}{ }\else{J. Stat.
  Phys.}\fi}\newcommand{\pnas }{\if\lName1\skp{ }{Proceedings of the National
  Academy of Sciences}{ }\else{PNAS}\fi}\newcommand{\lncs }{\if\lName1\skp{
  }{Lecture Notes in Computer Science}{ }\else{L. Notes Comp.
  Sci.}\fi}\newcommand{\lnai }{\if\lName1\skp{ }{Lecture Notes in Artificial
  Intelligence}{ }\else{L. Notes Art. Int.}\fi}\newcommand{\lnm
  }{\if\lName1\skp{ }{Lecture Notes in Mathematics}{ }\else{L. Notes
  Math.}\fi}\newcommand{\tams }{\if\lName1\skp{ }{Transactions of the American
  Mathematical Society}{ }\else{Trans. AMS}\fi}\newcommand{\ieeetit
  }{\if\lName1\skp{ }{{IEEE} Transactions on Information Theory}{ }\else{{IEEE}
  Trans. Inf. Th.}\fi}\newcommand{\iscs }{\if\lName1\skp{ }{International
  Series in Computer Science}{ }\else{Int. Ser. Comp.
  Sci.}\fi}\newcommand{\tocl }{\if\lName1\skp{ }{Theory of Computing Library}{
  }\else{Th. Comp. Lib.}\fi}
\begin{thebibliography}{DMWL21}

\bibitem[CDG{\etalchar{+}}20]{chao2020FindingAngleSequences}
Rui Chao, Dawei Ding, András Gilyén, Cupjin Huang, and Márió Szegedy.
\newblock Finding angles for quantum signal processing with machine precision.
\newblock \arxiv{2003.02831}, 2020.

\bibitem[CGJ19]{chakraborty2018BlockMatrixPowers}
Shantanav Chakraborty, András Gilyén, and Stacey Jeffery.
\newblock \href{http://dx.doi.org/10.4230/LIPIcs.ICALP.2019.33}{The power of
  block-encoded matrix powers: {I}mproved regression techniques via faster
  {H}amiltonian simulation}.
\newblock In {\em \icalp{46th}}, pages 33:1--33:14, 2019.
\newblock \arxiv{1804.01973}.

\bibitem[CW12]{childs2012HamSimLCU}
Andrew~M. Childs and Nathan Wiebe.
\newblock \href{http://dx.doi.org/10.26421/QIC12.11-12}{Hamiltonian simulation
  using linear combinations of unitary operations}.
\newblock {\em \qic}, 12(11\&12):901--924, 2012.
\newblock \arxiv{1202.5822}.

\bibitem[DMB{\etalchar{+}}23]{dalzell2023QuantumAlgSurvey}
Alexander~M. Dalzell, Sam McArdle, Mario Berta, Przemyslaw Bienias, Chi-Fang
  Chen, András Gilyén, Connor~T. Hann, Michael~J. Kastoryano, Emil~T.
  Khabiboulline, Aleksander Kubica, Grant Salton, Samson Wang, and Fernando G.
  S.~L. Brandão.
\newblock Quantum algorithms: A survey of applications and end-to-end
  complexities.
\newblock \arxiv{2310.03011}, 2023.

\bibitem[DMWL21]{dong2020efficientPhaseFindingInQSP}
Yulong Dong, Xiang Meng, K.~Birgitta Whaley, and Lin Lin.
\newblock \href{http://dx.doi.org/10.1103/PhysRevA.103.042419}{Efficient
  phase-factor evaluation in quantum signal processing}.
\newblock {\em \pra}, 103(4):042419, 2021.
\newblock \arxiv{2002.11649}.

\bibitem[git23]{gitHubRepo2023}
Source code available at \url{https://github.com/blankino/MQSP-examples}, 2023.

\bibitem[GSLW19]{gilyen2018QSingValTransf}
András Gilyén, Yuan Su, Guang~Hao Low, and Nathan Wiebe.
\newblock \href{http://dx.doi.org/10.1145/3313276.3316366}{Quantum singular
  value transformation and beyond: {E}xponential improvements for quantum
  matrix arithmetics}.
\newblock In {\em \stoc{51st}}, pages 193--204, 2019.
\newblock \arxiv{1806.01838}.

\bibitem[Haa19]{haah2018ProdDecPerFuncQSignPRoc}
Jeongwan Haah.
\newblock \href{http://dx.doi.org/10.22331/q-2019-10-07-190}{Product
  {D}ecomposition of {P}eriodic {F}unctions in {Q}uantum {S}ignal
  {P}rocessing}.
\newblock {\em \quantum}, 3:190, 2019.
\newblock \arxiv{1806.10236}.

\bibitem[LC17a]{low2017HamSimUnifAmp}
Guang~Hao Low and Isaac~L. Chuang.
\newblock Hamiltonian simulation by uniform spectral amplification.
\newblock \arxiv{1707.05391}, 2017.

\bibitem[LC17b]{low2016HamSimQSignProc}
Guang~Hao Low and Isaac~L. Chuang.
\newblock \href{http://dx.doi.org/10.1103/PhysRevLett.118.010501}{Optimal
  {H}amiltonian simulation by quantum signal processing}.
\newblock {\em \prl}, 118(1):010501, 2017.
\newblock \arxiv{1606.02685}.

\bibitem[LC19]{low2016HamSimQubitization}
Guang~Hao Low and Isaac~L. Chuang.
\newblock \href{http://dx.doi.org/10.22331/q-2019-07-12-163}{Hamiltonian
  simulation by qubitization}.
\newblock {\em \quantum}, 3:163, 2019.
\newblock \arxiv{1610.06546}.

\bibitem[LYC16]{low2016CompositeQuantGates}
Guang~Hao Low, Theodore~J. Yoder, and Isaac~L. Chuang.
\newblock \href{http://dx.doi.org/10.1103/PhysRevX.6.041067}{Methodology of
  resonant equiangular composite quantum gates}.
\newblock {\em \prx}, 6(4):041067, 2016.
\newblock \arxiv{1603.03996}.

\bibitem[MFM23]{mori2023}
Hitomi Mori, Keisuke Fujii, and Kaoru Mizuta.
\newblock Comment on "{M}ultivariable quantum signal processing ({M}-{QSP}):
  prophecies of the two-headed oracle".
\newblock \arxiv{2310.00918}, 2023.

\bibitem[MW23]{motlagh2023GeneralQSP}
Danial Motlagh and Nathan Wiebe.
\newblock Generalized quantum signal processing.
\newblock \arxiv{2308.01501}, 2023.

\bibitem[RC22]{rossi2022multivariableQSP}
Zane~M. Rossi and Isaac~L. Chuang.
\newblock \href{http://dx.doi.org/10.22331/q-2022-09-20-811}{Multivariable
  quantum signal processing ({M}-{QSP}): prophecies of the two-headed oracle}.
\newblock {\em \quantum}, 6:811, 2022.
\newblock \arxiv{2205.06261}.

\bibitem[WDL21]{wang2021energyLandscapeSymmetricQSP}
Jiasu Wang, Yulong Dong, and Lin Lin.
\newblock On the energy landscape of symmetric quantum signal processing.
\newblock \arxiv{2110.04993}, 2021.

\end{thebibliography}

\appendix

\section{Numerical search for counterexamples}\label{apx:program}

To develop an intuition on the veracity of \cite[Conjecture~2.1]{rossi2022multivariableQSP}, we built a program in Python designed to find potential counterexamples in the bivariate case for which a product decomposition does not exist. It was through this program that we first found some approximate counterexamples before arriving at the explicit formulae in \autoref{main}.

Our code~\cite{gitHubRepo2023} consists of two main parts: first we use gradient descent to find a matrix-valued bivariate Laurent polynomial $F$ such that $F(a,b) \in SU(2)$ on the torus. Then we reduce the problem to the univariate case, and implement Haah’s univariate product decomposition.

First we want to find $F(a,b)$ such that $F(a,b) \in SU(2)$ for all $(a,b)\in \mathrm{U(1)}\times\mathrm{U(1)}$. If $F(a,b) \in SU(2)$, then it is of the form
$$F(a,b)=
\begin{bmatrix}
P(a,b) & Q(a,b)\\
-Q^*\left(a^{-1},b^{-1}\right) & P^*\left(a^{-1},b^{-1}\right)
\end{bmatrix}
$$
for some Laurent polynomials $P(a,b)$ and $Q(a,b)$. $F(a,b) \in SU(2)$ also implies 
$\left|P(a,b)\right|^2+\left|Q(a,b)\right|^2=1$ for $(a,b)\in \mathrm{U(1)}\times\mathrm{U(1)}$. Hence we need to find coefficients for $P$ and $Q$ such that
\[
f(P,Q)=\int_{\theta=0}^{2\pi}\int_{\varphi=0}^{2\pi}\left(1-\left| P\left(e^{i\theta},e^{i\varphi}\right) \right| ^2  -\left|Q\left(e^{i\theta},e^{i\varphi}\right)\right| ^2\right)^2d\varphi d\theta= 0
\] on the torus $\mathbb{T}^2=[0,2\pi)\times [0,2\pi)$. The function $f$ is a positive semi-definite quadratic form in the coefficients of the polynomials $P$ and $Q$, and it is zero if and only if $P(a,b)$ and $Q(a,b)$ are polynomials such that $F(a,b)\in\mathrm{SU}(2)$ for all $(a,b)\in\mathbb T^2$. Therefore we may use gradient descent to find the roots of $f$. The evaluation of the integral can be done efficiently using trigonometric quadrature: if the degrees of $P$ and $Q$ in $a$ and $b$ are bounded above by $N$ then $f$ can be exactly calculated by trapezoidal integration on a grid of size $(2N+1)\times (2N+1)$ of roots of unity, that is, of pairs of complex numbers of the form $\left(e^{i\frac{2\pi k}{2N+1}},e^{i\frac{2\pi l}{2N+1}}\right)$, for $(k,l)\in \{0,\dots,2N\}^2$. We initialize the gradient search from random initial coefficients.   

Now substituting $t$ for both $a$ and $b$ in $F$, we arrive at a matrix $F(t)$ which can be uniquely decomposed by Haah. If there exists a bivariate product decomposition of $F(a,b)$, it must be the case that we can obtain this decomposition by substituting $a$ or $b$ into each matrix $E_{P_k}(t)$ of the univariate Haah decomposition. The number of $a$’s and $b$’s is determined by the degree of $F$ in $a$ and $b$. We test all possible permutations of $a$’s and $b$’s, and if none of these substitutions yield a matrix ‘close enough’ to $F$ (with respect to the precision of the gradient descent), then we conclude that a decomposition is impossible.

In the case $d_A=d_B=2$, out of the $1934$ polynomials we generated, approximately $\approx 18.5 \%$ did not have a product decomposition. If a decomposition did exist, it was very likely to belong to permutations either ‘aabb’ or ‘bbaa’, both accounting for $\approx 45 \%$ of decompositions.

\section{A counterexample with rational coefficients}

We used the following Mathematica code to find and verify our rational counterexample.

\mmaDefineMathReplacement{α}{\alpha}
\mmaDefineMathReplacement{β}{\beta}
\mmaDefineMathReplacement{γ}{\gamma}
\begin{mmaCell}[functionlocal=α1,mathreplacements=bold]{Code}
\mmaUnd{γ0} = \mmaUnd{α0}*k*I;
\mmaUnd{γ1} = \mmaUnd{α1}/(k*I);
AA = 1;
FF = 1;
BB = -(\mmaUnd{α0}+\mmaUnd{α1});
CC = \mmaUnd{α0}*\mmaUnd{α1};
DD = -(\mmaUnd{γ0}+\mmaUnd{γ1});
GG = -(\mmaUnd{α0}+Conjugate[\mmaUnd{α1}]);
HH = \mmaUnd{α0}*Conjugate[\mmaUnd{α1}];
II = -(\mmaUnd{γ0}+Conjugate[\mmaUnd{γ1}]);
r = Conjugate[GG]*II+Conjugate[II]*GG-BB*Conjugate[DD]-DD*Conjugate[BB];
s = -(Conjugate[GG]*II+GG*Conjugate[II]+BB*Conjugate[DD]+DD*Conjugate[BB]);
t = -(Conjugate[AA]*BB+Conjugate[BB]*AA+BB*Conjugate[CC]+Conjugate[BB]*CC+
    +FF*Conjugate[GG]+GG*Conjugate[FF]+GG*Conjugate[HH]+Conjugate[GG]*HH);
u = II*Conjugate[HH]+HH*Conjugate[II]-FF*Conjugate[II]-Conjugate[FF]*II+
    -AA*Conjugate[DD]-DD*Conjugate[AA]-CC*Conjugate[DD]-Conjugate[CC]*DD;
x = (DD*Conjugate[AA]-AA*Conjugate[DD])*(r*CC-s*AA)+
    -(r*DD-t*AA)*(Conjugate[AA]*CC-Conjugate[CC]*AA);
y = (r*DD-t*AA)*(BB-Conjugate[BB])-(r*BB-u*AA)*(DD-Conjugate[DD]);
Solve[ x==0 && y==0 /. {k -> 3, \mmaUnd{α0} -> 1 + I}, \mmaFnc{α1}, Complexes]

\end{mmaCell}

\begin{mmaCell}{Output}
\Bigg\{\Big\{\mmaUnd{α1} -> \mmaFrac{85}{37} - \mmaFrac{29 I}{37}\Big\}, \Big\{\mmaUnd{α1} -> 9 - 10 I\Big\}\Bigg\}
\end{mmaCell}
\end{document}